\documentclass[11pt,epsf]{article}

\usepackage{complexity}

\usepackage{tcolorbox}
\usepackage[T1]{fontenc}
\usepackage{mathpazo}
\usepackage{graphicx}
\usepackage{dirtytalk}
\usepackage{bbold}
\usepackage{caption}
\DeclareCaptionLabelFormat{nolabel}{}
\usepackage{chngcntr} 
\usepackage{enumerate} 
\usepackage[
  margin=2cm,
  includefoot,
  footskip=30pt,
]{geometry}
\usepackage{amsmath} 
\usepackage{amssymb} 
\usepackage{textcomp}
    \AtBeginDocument{%
        
    }
\usepackage{upquote} 
\usepackage{eurosym} 
\usepackage[mathletters]{ucs} 
\usepackage[utf8]{inputenc} 
\usepackage{fancyvrb} 
\usepackage{grffile} 
\usepackage{longtable} 
\usepackage{booktabs} 
\usepackage[inline]{enumitem} 
\usepackage[normalem]{ulem} 

\usepackage[pdftex,pagebackref,colorlinks,linkcolor=blue,filecolor = blue, citecolor = magenta, urlcolor  = JungleGreen]{hyperref}

    \definecolor{urlcolor}{rgb}{0,.145,.698}
    \definecolor{linkcolor}{rgb}{.71,0.21,0.01}
    \definecolor{citecolor}{rgb}{.12,.54,.11}

    \definecolor{ansi-black}{HTML}{3E424D}
    \definecolor{ansi-black-intense}{HTML}{282C36}
    \definecolor{ansi-red}{HTML}{E75C58}
    \definecolor{ansi-red-intense}{HTML}{B22B31}
    \definecolor{ansi-green}{HTML}{00A250}
    \definecolor{ansi-green-intense}{HTML}{007427}
    \definecolor{ansi-yellow}{HTML}{DDB62B}
    \definecolor{ansi-yellow-intense}{HTML}{B27D12}
    \definecolor{ansi-blue}{HTML}{208FFB}
    \definecolor{ansi-blue-intense}{HTML}{0065CA}
    \definecolor{ansi-magenta}{HTML}{D160C4}
    \definecolor{ansi-magenta-intense}{HTML}{A03196}
    \definecolor{ansi-cyan}{HTML}{60C6C8}
    \definecolor{ansi-cyan-intense}{HTML}{258F8F}
    \definecolor{ansi-white}{HTML}{C5C1B4}
    \definecolor{ansi-white-intense}{HTML}{A1A6B2}

    \DefineVerbatimEnvironment{Highlighting}{Verbatim}{commandchars=\\\{\}}

    \definecolor{incolor}{rgb}{0.0, 0.0, 0.5}
    \definecolor{outcolor}{rgb}{0.545, 0.0, 0.0}

    \hypersetup{
      breaklinks=true,  
      colorlinks=true,
      urlcolor=urlcolor,
      linkcolor=linkcolor,
      citecolor=citecolor,
      }
    
    \usepackage{pstricks}

\usepackage{amsmath,amstext,amssymb,amsfonts}
\usepackage{graphicx}
\usepackage{xcolor}
\usepackage{tabularx}
\usepackage{verbatim}
\usepackage{fullpage}

\usepackage{fourier}

\usepackage[T1]{fontenc}
\usepackage{bbm}

\usepackage{amsthm}
\usepackage{thmtools, thm-restate}

\usepackage{ifdraft}
\usepackage{appendix}
\usepackage{multicol}

\usepackage{nicefrac}

\usepackage{microtype}

\usepackage{algorithm,algorithmicx}

\newcommand{\gcap}[1]{\bar{\Phi}\paren{#1}}
\newcommand{\ggcap}[1]{{\Phi}\paren{#1}}

\newcommand{\vect}[1]{{\sf v}_{#1}}
\newcommand{\vectu}[1]{{\sf u}_{#1}}
\newcommand{\vectone}{\vect{\emptyset}}

\theoremstyle{plain}
\newtheorem{theorem}{Theorem}
\newtheorem{proposition}[theorem]{Proposition}
\newtheorem{sdp}[theorem]{SDP}
\newtheorem{lemma}[theorem]{Lemma}
\newtheorem{claim}[theorem]{Claim}
\newtheorem{corollary}[theorem]{Corollary}
\newtheorem{definition}[theorem]{Definition}
\newtheorem{fact}[theorem]{Fact}

\newtheorem{observation}[theorem]{Observation}

\usepackage{dirtytalk}
\usepackage{bbold}

\usepackage{enumerate} 

\usepackage{amsmath} 
\usepackage{amssymb}

\usepackage[mathletters]{ucs}

    \definecolor{urlcolor}{rgb}{0,.145,.698}
    \definecolor{linkcolor}{rgb}{.71,0.21,0.01}
    \definecolor{citecolor}{rgb}{.12,.54,.11}

\usepackage{amsmath,amstext,amssymb,amsfonts}
\usepackage{graphicx}
\usepackage{xcolor}
\usepackage{tabularx}
\usepackage{verbatim}
\usepackage{bbm}

\usepackage{microtype}

\usepackage{prettyref}
\newcommand{\savehyperref}[2]{\texorpdfstring{\hyperref[#1]{#2}}{#2}}

\newrefformat{parta}{\savehyperref{#1}{a}}
\newrefformat{partb}{\savehyperref{#1}{b}}
\newrefformat{eq}{\savehyperref{#1}{\textup{(\ref*{#1})}}}
\newrefformat{lem}{\savehyperref{#1}{Lemma~\ref*{#1}}}
\newrefformat{def}{\savehyperref{#1}{Definition~\ref*{#1}}}
\newrefformat{thm}{\savehyperref{#1}{Theorem~\ref*{#1}}}
\newrefformat{cor}{\savehyperref{#1}{Corollary~\ref*{#1}}}
\newrefformat{cha}{\savehyperref{#1}{Chapter~\ref*{#1}}}
\newrefformat{sec}{\savehyperref{#1}{Section~\ref*{#1}}}
\newrefformat{app}{\savehyperref{#1}{Appendix~\ref*{#1}}}
\newrefformat{tab}{\savehyperref{#1}{Table~\ref*{#1}}}
\newrefformat{fig}{\savehyperref{#1}{Figure~\ref*{#1}}}
\newrefformat{hyp}{\savehyperref{#1}{Hypothesis~\ref*{#1}}}
\newrefformat{alg}{\savehyperref{#1}{Algorithm~\ref*{#1}}}
\newrefformat{sdp}{\savehyperref{#1}{SDP~\ref*{#1}}}
\newrefformat{qp}{\savehyperref{#1}{QP~\ref*{#1}}}
\newrefformat{vp}{\savehyperref{#1}{VP~\ref*{#1}}}
\newrefformat{lp}{\savehyperref{#1}{LP~\ref*{#1}}}
\newrefformat{rem}{\savehyperref{#1}{Remark~\ref*{#1}}}
\newrefformat{item}{\savehyperref{#1}{Item~\ref*{#1}}}
\newrefformat{step}{\savehyperref{#1}{step~\ref*{#1}}}
\newrefformat{conj}{\savehyperref{#1}{Conjecture~\ref*{#1}}}
\newrefformat{fact}{\savehyperref{#1}{Fact~\ref*{#1}}}
\newrefformat{prop}{\savehyperref{#1}{Proposition~\ref*{#1}}}
\newrefformat{claim}{\savehyperref{#1}{Claim~\ref*{#1}}}
\newrefformat{relax}{\savehyperref{#1}{Relaxation~\ref*{#1}}}
\newrefformat{red}{\savehyperref{#1}{Reduction~\ref*{#1}}}
\newrefformat{part}{\savehyperref{#1}{Part~\ref*{#1}}}
\newrefformat{prob}{\savehyperref{#1}{Problem~\ref*{#1}}}
\newrefformat{ass}{\savehyperref{#1}{Assumption~\ref*{#1}}}
\newrefformat{cons}{\savehyperref{#1}{Construction~\ref*{#1}}}
\newrefformat{obs}{\savehyperref{#1}{Observation~\ref*{#1}}}
\newrefformat{step}{\savehyperref{#1}{Step~\ref*{#1}}}
\newrefformat{item}{\savehyperref{#1}{Item~\ref*{#1}}}

\newcommand{\Sref}[1]{\hyperref[#1]{\S\ref*{#1}}}

\renewcommand{\leq}{\leqslant}

\renewcommand{\geq}{\geqslant}

\usepackage{bm}

\newcommand{\paren}[1]{\left(#1 \right )}

\newcommand{\set}[1]{\left\{#1\right\}}

\newcommand{\abs}[1]{\left\lvert#1\right\rvert}
\newcommand{\Abs}[1]{\left\lvert#1\right\rvert}

\newcommand{\floor}[1]{\left\lfloor #1 \right\rfloor}

\newcommand{\norm}[1]{\left\lVert#1\right\rVert}

\newcommand{\defeq}{\stackrel{\textup{def}}{=}}
\newcommand{\assign}{:=}

\newcommand{\inprod}[1]{\left\langle #1\right\rangle}

\newcommand{\card}{\abs}

\DeclareMathOperator*{\ProbOp}{\sf Pr}

\newcommand{\Prob}[2]{\underset{#1}{\ProbOp}\left[#2\right]}

\newcommand{\Ex}[2]{\underset{#1}{\E}\left[#2\right]}

\newcommand{\eps}{\varepsilon}

\definecolor{DSgray}{cmyk}{0,0,0,0.7}

\newcommand{\cN}{N}

\title{Improved linearly ordered colorings of hypergraphs via SDP rounding}

\author{
	Anand Louis
	\footnote{Supported in part by SERB Award CRG/2023/002896 and the Walmart Center for Tech Excellence at IISc (CSR Grant WMGT-23-0001).}\\
	Indian Institute of Science, Bengaluru\\
	\href{mailto:anandl@iisc.ac.in}{anandl@iisc.ac.in}
	\and
	Alantha Newman\\
	CNRS and Universit\'e Grenoble Alpes\\
	\href{mailto:alantha.newman@grenoble-inp.fr}{alantha.newman@grenoble-inp.fr}
	\and
	Arka Ray\footnote{Supported in part by the Walmart Center for Tech Excellence at IISc (CSR Grant WMGT-23-0001).}\\
	Indian Institute of Science, Bengaluru\\
	\href{mailto:arkaray@iisc.ac.in}{arkaray@iisc.ac.in}
	}
\date{}

\begin{document}
\maketitle

\begin{abstract}
	
We consider the problem of {\em linearly ordered} (LO) coloring of
hypergraphs.  A hypergraph has an LO coloring if there is a vertex
coloring, using a set of ordered colors, so that (i) no edge is
monochromatic, and (ii) each edge has a unique maximum color.  It is
an open question as to whether or not a 2-LO colorable 3-uniform
hypergraph can be LO colored with 3 colors in polynomial time.
Nakajima and \v{Z}ivn{\'{y}} recently gave a polynomial-time algorithm to
color such hypergraphs with $\widetilde{O}(n^{1/3})$ colors and asked
if SDP methods can be used directly to obtain improved bounds.  Our
main result is to show how to use SDP-based rounding methods to
produce an LO coloring with $\widetilde{O}(n^{1/5})$ colors for such
hypergraphs.  We show how to reduce the problem to cases
with highly structured SDP solutions, which we call {\em balanced}
hypergraphs.  Then we discuss how to apply classic SDP-rounding tools in
this case to obtain improved bounds.  



\end{abstract}


\section{Introduction}

Approximate graph coloring is a well-studied ``promise'' optimization
problem.  Given a simple graph $G=(V,E)$ that is promised to be
$k$-colorable, the goal is to find a coloring of $G$ using the minimum
number of colors.  A (proper) {\em coloring} is an assignment of
colors, which can be represented by positive integers, to the vertices
of $G$ so that for each edge $ij$ in $G$, the vertices $i$ and $j$ are
assigned different colors.  The most popular case of this problem is
when the input graph is promised to be $3$-colorable.  Even with this
very strong promise, the gap between the upper and lower bounds are
quite large: the number of colors used by the state-of-the-art
algorithm is $\widetilde O\paren{n^{0.19996}}$~\cite{KT17}, while it
is $\NP$-hard to color a 3-colorable graph with 5
colors~\cite{BBKO21}.  There is also super constant hardness
conditioned on assumptions related to the Unique Games
Conjecture~\cite{dinur2006conditional}.  More generally, when we are
promised that the graph $G$ is $k$-colorable, it is NP-hard to color
it using $\binom{k}{\floor{k/2}}-1$ colors~\cite{WZ20}.  Regarding
upper bounds, we note that almost all algorithms for coloring 3-colorable
graphs use some combination of semidefinite programming (SDP) and
combinatorial tools~\cite{KMS98,arora2006new,KT17}.

Approximate hypergraph coloring is a natural generalization of the
above problem to hypergraphs.  Here, we want to assign each vertex a
color such that there are no monochromatic edges, while using the
minimum number of colors.  In the case of hypergraph coloring, we know
that for every pair of constants $\ell \geq k\geq 2$, it is NP-hard to
$\ell$-color a $k$-colorable 3-uniform hypergraph~\cite{DRS05}.  Even
in the special case, when the 3-uniform hypergraph is promised to be
2-colorable, there is a large gap between the best algorithm, which
uses at most $\widetilde O\paren{n^{1/5}}$ 
colors~\cite{KNS01,alon1996coloring,chen1996coloring} and the
aforementioned (super constant) lower bound.

In this paper, we study a variant of the hypergraph coloring problem
known as {\em linearly ordered coloring}, introduced in several
different contexts by
\cite{katchalski1995ordered,cheilaris2011graph,BBB21}.  A linearly
ordered (LO) $k$-coloring of an $r$-uniform hypergraph assigns an
integer from $\set{1, \dots , k}$ to every vertex so that, in each
edge in the hypergraph, there is a unique vertex assigned the maximum color in the 
(multi)set of colors for that edge.  Recently, there
has been a renewed interest in studying this problem.  This is because
this problem constitutes a gap in the understanding of the complexity
of an important class of problems called {\em promise constraint
  satisfaction problems} (PCSPs).  To elaborate, ~\cite{FKOS19, BG21}
classified the complexity of all (symmetric) PCSPs on the binary
alphabet, showing that these problems are either polynomial-time
solvable or \NP-complete.  Subsequently, \cite{BBB21} gave a complete
classification for PCSPs of the form: given a 2-colorable 3-uniform
hypergraph, find a 3-coloring.  Here, the notion of ``coloring'' can
have several definitions.  As highlighted by \cite{BBB21}, the only
PCSP of this type whose complexity is unresolved is that of
determining whether a 3-uniform hypergraph is 2-LO colorable or is not
even 3-LO colorable.  In contrast, it was recently shown that it is
\NP-complete to decide if a 3-uniform hypergraph is 3-LO colorable or
not even 4-LO colorable~\cite{FNOTW24}.

The work~\cite{NZ22} addresses the corresponding optimization problem
by giving an algorithm to compute an LO coloring using at most
$\widetilde O\paren{n^{1/3}}$ colors for a 2-LO colorable 3-uniform
hypergraph. \cite{NZ22} leave open the question of finding an LO
coloring for such a hypergraph using fewer colors.  Moreover, they
state that they do not know how to directly use SDP-based
methods\footnote{However, they do use~\cite{Hal00} which is an
  indirect use of SDP-based methods.} and remark that SDP-based
approaches seem ``less suited for LO colorings''.  In this paper, one of
our main contributions is to show how to use SDP relaxations to give an
improved bound for coloring such hypergraphs.  Our main result
improves this bound significantly by using at most $\widetilde O
\paren{n^{1/5}}$ colors to LO color a 2-LO colorable 3-uniform
hypergraph.
\begin{restatable}{theorem}{loColor}
\label{thm:lo-color}
Let $H$ be a 2-LO colorable 3-uniform hypergraph on $n$ vertices.
Then there exists a (randomized) polynomial-time algorithm that finds
an LO coloring of $H$ using \ $\widetilde O\paren{n^{1/5}}$ colors.
\end{restatable}

The SDP relaxation that we use is similar to the natural SDP used in
the case of 2-colorable 3-uniform hypergraphs~\cite{KNS01}. In fact,
the upper bound on the number of colors used in
\prettyref{thm:lo-color} is the same as the upper bound given by
\cite{KNS01} to color 2-colorable 3-uniform hypergraphs.  It is the
same SDP used by \cite{brakensiek2023sdps} who show that a
straightforward hyperplane rounding algorithm yields a solution
to the PCSP {\bf{(\sc{1-in-3-SAT}, \sc{NAE-3-SAT})}}, in which
we are given a satisfiable instance of the first problem and we want to find a
feasible solution for the second.  Notice that a satisfiable
{\bf{(\sc{1-in-3-SAT})}} instance on all positive literals is exactly
a 2-LO colorable 3-uniform hypergraph.


\paragraph*{General Framework for (Hyper)Graph Coloring.}
Most algorithms for coloring graphs and hypergraphs proceed
iteratively, producing a partial coloring of the remaining (uncolored)
vertices at each step.  This was formalized by \cite{blum1994new},
following \cite{Wig83}.  The goal is to color a significant number of
vertices with few colors in each step, ensuring that the number of
iterations and therefore, the overall number of colors used, is small.
Typically, in each step, the method used to color the vertices is
chosen according to the degree of the graph (or hypergraph) induced on
the remaining vertices.  In particular, if the induced graph (or
hypergraph) has a low degree, then most algorithms use an SDP-based
method to find a large independent set, which can be assigned a single
color~\cite{KMS98,blum1997n,arora2006new}.  The algorithm for
LO-coloring presented in \cite{NZ22}, as well as ours, uses this
general framework, except that in \cite{NZ22}, they did not use an
SDP-based method directly, and instead used~\cite{Hal00} to find a large
independent set.  The improved upper bound on the number of colors
output by our algorithm comes from using an SDP and rounding methods
tailored to LO coloring.

\paragraph*{Overview of our SDP-Based Approach.}
As noted, we first solve a natural SDP relaxation for 2-LO coloring.
Then our rounding proceeds in two steps.  In the first step we look at
the projection of the vectors to a particular special vector (the
vector $\vect{\emptyset}$ in \autoref{sdp:lo_coloring}) from the
solution of the SDP, which signifies the color that is unique in all
edges in the promised 2-LO coloring.  For each of the three vertices
in an edge, all three of the corresponding vectors can have a
projection onto this special vector with roughly the same value (a
{\em balanced} edge), or they can have very different values (an {\em
  unbalanced} edge).  It is also possible to classify vertices into
balanced and unbalanced (see \autoref{def:balanced} for formal
definitions) so that balanced edges contain only balanced vertices.
We use a combinatorial rounding procedure to color all the unbalanced
vertices with a small number of colors, leaving only a balanced
(sub)hypergraph to be colored. Since this number of colors is much
smaller than the bound stated in \prettyref{thm:lo-color}, this can be
viewed as a reduction of the problem to the balanced case.  To the
best of our knowledge, this rounding method is not present in previous
works on LO-coloring and thus, this tool can be considered a main
contribution of this paper.  We note that~\cite{KNS01} showed that the
vectors can be ``bucketed'' with respect to their projection onto a
special vector, and used a simple argument to show that there is a
large bucket on which they can focus. Our approach allows us to focus
on a single bucket containing vectors with projection $\approx -1/3$
with the special vector, which have useful geometric properties.

In the second step we color the hypergraph containing the balanced
edges.  In this step, we produce (following \cite{NZ22}) an ``even''
independent set or an ``odd'' independent set at each round.  An {\em
  even independent set} is one which intersects each hyperedge two or
zero times, while an {\em odd independent set} intersects each
hyperedge one or zero times.  To find an even independent set, we use
the same approach used by \cite{NZ22}.  To find an odd independent
set, we use a variant of the standard threshold rounding for a
coloring SDP~\cite{Hal00, KNS01}. As in~\cite{KNS01} rather than use
the vectors output by the SDP solution, we use a modified set of
vectors, which have properties useful to obtain better bounds from the
threshold rounding.  Specifically, the set consists of the normalized
projections of the vectors from the SDP solution onto the space
orthogonal to the special vector; in the balanced case, the special
vector seems to provide no information that is useful to construct a
coloring.  Combining all the colorings requires some technical care,
since we need to always maintain an LO coloring, but it can be done
and some of the work has already been done in \cite{NZ22}.


\paragraph*{Update on Independent and Subsequent Work.}  After the
initial conference submission of our paper, the work \cite{hmnz24}
appeared on the arXiv.  The second version appeared after we posted
our paper to arXiv and pointed out that in fact we do not need to consider
the balanced case.  Indeed, the observation in Section 3 of
\cite{hmnz24} can be interpreted as giving an alternative and better
SDP rounding in the balanced case, directly reducing the balanced case
to the unbalanced case.  We discuss this more at the end of \prettyref{sec:sdpRound}.

\section{Tools for LO Coloring and Proof of the Main Theorem} \label{sec:overview}

In this section, we give an overview of our approach to color a 2-LO
colorable 3-uniform hypergraph $H=(V,E)$ with few colors.  Following
\cite{NZ22}, we assume that the input hypergraph $H$ is a {\em linear
  hypergraph}, which is defined as follows.
\begin{definition}
A 3-uniform hypergraph is
linear if every pair of edges intersects in at most one vertex.
\end{definition}
This is not a restriction because we can construct an equivalent
3-uniform hypergraph.
\begin{proposition}[Proposition 3 in \cite{NZ22}]
\label{prop:linear-reduction}
There is a polynomial-time algorithm that, if given an 2-LO colorable
3- uniform hypergraph $H$, constructs an 2-LO colorable linear
3-uniform hypergraph $H'$ with no more vertices than $H$ such that, if
given an LO $k$-colouring of $H'$, one can compute in polynomial time
an LO $k$-colouring of H.
\end{proposition}

Given an 2-LO colorable 3-uniform hypergraph $H=(V,E)$, one can
consider LO coloring it with $\set{-1, +1}$, with the natural
ordering.  Then we have $x_a + x_b + x_c = -1$ for each edge $\set{a,
  b, c}\in E$, where $x_a$ is the color assigned to vertex $a\in V$.
Relaxing this constraint to a vector program we get
\prettyref{sdp:lo_coloring}.\footnote{Observe that
\prettyref{sdp:lo_coloring} can equivalently be written in terms of
dot products using the following constraints:\\ (i) $\inprod{\vect{a} +
\vect{b} + \vect{c} + \vect{\emptyset}, \vect{a} + \vect{b} + \vect{c} + \vect{\emptyset}} = 0 \quad
\forall \{a,b,c\} \in E$, \quad and \quad (ii) $ \inprod{\vect{a},
\vect{a}} = 1 \quad \forall a \in V\cup\set{\emptyset}$.  }
\begin{sdp}
\label{sdp:lo_coloring}
\begin{align}
\vect{a} + \vect{b} + \vect{c} & = - \vectone & \forall \set{a,b,c} \in E, \label{eq:sdp-sum} \\ 
\norm{\vect{a}}^2 & = 1 & \forall a \in V \cup \set{\emptyset}. 
\end{align}
\end{sdp}

For any $a \in V$, we now define $\gamma_a \defeq
\inprod{\vect{a}, \vectone}$.  The values $\set{\gamma_a}_{a\in V}$ might not be
integral and could even be \emph{perfectly balanced} (i.e.,
$\gamma_a=\gamma_b=\gamma_c=-\frac{1}{3}$ for an edge $\{a,b,c\} \in
E$).  Hence, these values might not contain any information as to how
the colors should be assigned to the vertices, and they might not even
reveal information as to which vertex in an edge should receive the
largest color.  However, when all edges contain balanced vertices
(i.e., $\gamma_v \approx -\frac{1}{3}$ for all vertices), threshold rounding will be used.  Formally, we have the following
definition.

\begin{definition}
\label{def:balanced}
For $\eps > 0$, we say a vertex $v \in V$ is {\em $\eps$-balanced} if
$\gamma_v\in [-1/3 - \eps, -1/3+\eps]$.
\end{definition}

For the rest of this paper, we fix $\eps = 1/n^{100}$, where $n$ is
the number of vertices in the (fixed) hypergraph that we are trying to
LO color.  This is an abuse of notation, but simplifies our
presentation.  If a vertex is not $\eps$-balanced, we say that it is
{\em unbalanced}.  If all vertices of a hypergraph $H$ are
$\eps$-balanced, we say that $H$ is an $\eps$-balanced hypergraph.

We observe that there is a combinatorial method to color all
unbalanced vertices using relatively few colors.  This rounding method
uses a bisection-like strategy on $\set{\gamma_a}_{a\in V}$ to color the
unbalanced vertices and outputs a {\em partial LO coloring}, which we
define as follows.
\begin{definition}
A {\em partial LO coloring} of a 3-uniform hypergraph $H=(V,E)$ is a coloring of a subset of vertices $V_1 \subseteq V$ using the set of colors $C$ such that for each edge $e \in E$, the set $e \cap V_1$ has a unique maximum color from $C$.
\end{definition}

The next lemma is proved in Section \ref{sec:combRound}.

\begin{restatable}{lemma}{unbalancedPartialChi}
  \label{lem:unbalancedPartialColoring}
Let $H = (V,E)$ be a 2-LO colorable 3-uniform hypergraph
and let $\eps > 0$.  Then there exists a polynomial-time algorithm
that computes a partial LO coloring of $H$ using
$O\paren{\log\paren{\frac{1}{\eps}}}$ colors that colors all
unbalanced vertices.
\end{restatable}

We remark that the previous lemma can be viewed as a reduction from LO
coloring in 2-LO colorable 3-uniform hypergraphs to LO coloring in
2-LO colorable 3-uniform {\em balanced} hypergraphs.  To formalize
this, let $V_U$ denote the vertices that are colored in a partial LO
coloring produced via \prettyref{lem:unbalancedPartialColoring}.
Let $V_B = V \setminus{V_U}$.
Notice that $V_B$ contains only
$\eps$-balanced vertices, while $V_U$ contains all the unbalanced
vertices but might also contain some $\eps$-balanced vertices.  Thus,
the induced hypergraph $H_B = (V_B, E(V_B))$ is a balanced
hypergraph.\footnote{Note that for a
  hypergraph $H=(V,E)$ and $S \subset V$, we say $H' = (S, E(S))$
  contains the edges {\em induced} on $S$, meaning an edge belongs to
  $H'$ if all of its vertices belong to $S$.  In other words, an
  induced subhypergraph of a 3-uniform must also be 3-uniform (or
  empty).  Notice that $S$ can contain vertices that do not belong to
  any edge in $E(S)$.  These vertices can receive any color in a valid
LO coloring of $H'$.}  
We now show that we
can combine a partial LO coloring for $H=(V,E)$ which colors $V_U$ and
an LO coloring for $H_B=(V_B,E(V_B))$ to obtain an LO coloring of $H$.

\begin{proposition}\label{prop:combine}
  Let $H=(V_B \cup V_U,E)$ be a 2-LO colorable 3-uniform hypergraphs,
  let $\eps > 0$.  Let $c_U$ be a partial LO coloring of $H$ using 
  colors from the set $C_U$ that only
  assigns colors to $V_U$ and let $c_B$ be an LO coloring of
  $H_B=(V_B,E(V_B))$ using colors from the set $C_B$.
  Then we can obtain an LO coloring of $H$ using at
  most $|C_U| + |C_B|$ colors.
\end{proposition}

\begin{proof}
We assume that the colors in the set $C_U$ are larger than the colors
in the set $C_B$.  We want to show that the given assignment of colors
from $C_U$ for vertex set $V_U$ and $C_B$ for vertex set $V_B$ taken
together forms a proper LO coloring of $H$.

Any edge $e \in E$ with $|e \cap V_B| = 3$ or $|e \cap V_U|=3$ has a
unique maximum color by assumption since $c_B$ is an LO coloring of
$H_B$ and $c_U$ is a partial LO coloring of $H$.  Suppose $|e \cap
V_U| = 2$.  Then, by definition of partial LO coloring, it has a unique
maximum in $C_U$ and will have a unique maximum in the output
coloring.  If $|e \cap V_U| = 1$, then $e$ has a unique maximum color,
because all colors in $C_U$ are larger than the colors in $C_B$.
  \end{proof}

Thus, if our goal is to LO color 2-LO colorable 3-uniform hypergraphs
with a polynomial number of colors, we can focus on LO coloring {\em
  balanced} 2-LO colorable 3-uniform hypergraphs.  The next
corollary follows from \prettyref{lem:unbalancedPartialColoring} and
\autoref{prop:combine}.
\begin{corollary}\label{cor:reduction}
Let $\alpha \in (0,1)$.
Suppose we can LO color an $\eps$-balanced 2-LO colorable 3-uniform
hypergraph $H$ with $\widetilde{O}(n^{\alpha})$ colors.  Then we can
LO color a 2-LO colorable 3-uniform hypergraph with $\widetilde{O}(n^{\alpha})$ colors. 
\end{corollary}

Now we can focus on balanced hypergraphs.  We capitalize on the
promised structure to prove the next lemma, in which we show that we
can find an LO coloring for a balanced hypergraph, in particular for $H_B =
(V_B, E(V_B))$.

\begin{lemma}\label{lem:colorbalanced}
Let $H_B = (V_B, E_B)$ be an $\eps$-balanced 2-LO colorable 3-uniform
hypergraph.  Then there exists a polynomial-time algorithm that
computes an LO coloring using at most \ $\widetilde{O}(|V_B|^{1/5})$ colors.  
  \end{lemma}

We recall our main theorem.

\loColor*

The proof of \prettyref{thm:lo-color} follows from
\prettyref{cor:reduction} and \autoref{lem:colorbalanced}.
It remains to prove \prettyref{lem:colorbalanced}, which we discuss next.

\subsection{Coloring by Finding Independent Sets}

In many graph coloring algorithms, we ``make progress'' by finding an
independent set and coloring it with a new
color~\cite{blum1994new,KMS98,blum1997n,KNS01,NZ22}.  When LO coloring
a hypergraph, a similar idea may be used, but we need to consider
certain types of independent sets.  With the standard notion of
independent set in a 3-uniform hypergraph, in which the independent
set intersects each edge of the hypergraph at most twice, it is not clear
how to obtain a coloring in which each edge contains a unique maximum
color.  Thus, for a 3-uniform hypergraph $H=(V,E)$, following the
approach of \cite{NZ22}, we consider the following two types of
independent sets.\footnote{We remark that what \cite{NZ22} refer to as an ``independent set'' is what we refer to here as an ``odd independent set''.}
\begin{description}
	\item[Odd Independent Set:] We call $S\subseteq V$ an \emph{odd independent set} if $|S\cap e| \leq  1$ for each edge $e\in E$.
	\item[Even Independent Set:] We call $S\subseteq V$ an \emph{even independent set} if $|S\cap e| \in \{0,2\}$ for each edge $e\in E$.  
\end{description}
In \prettyref{lem:lo-valid-progress}, we show that we can make
progress by coloring an odd independent set with a `large' color or by
coloring an even independent set with a `small' color.  This is
formally stated in a proposition from \cite{NZ22}.  Since we modify the
presentation slightly to ensure compatibility with our framework, we 
include the
statement and the proof here for the sake of completeness.

\begin{lemma}[Corollary of Proposition 5 in \cite{NZ22}]
  \label{lem:lo-valid-progress}
  Let $H=(V,E)$ be a hypergraph, let $S_1\subseteq V$ be an odd
  independent set and let $S_2\subseteq V$ be an even independent set.
  Let $H_1=(V_1, E_1)$, $H_2=(V_2, E_2)$ be the hypergraphs induced by
  $V_1=V\setminus S_1$ and $V_2=V\setminus S_2$, respectively.  Then,
\begin{enumerate}
\item An LO coloring of $H_1$ using a set of colors $C_1$ can be extended to an LO coloring of $H$ by assigning a color $c_1$ that is strictly larger than all the colors in $C_1$ to the vertices in $S_1$.

\item Analogously, an LO coloring of $H_2$ using a set of colors $C_2$ can be extended to an LO coloring of $H$ by assigning $c_2$ to the vertices in $S_2$ where $c_2$ is strictly smaller than all the colors in $C_2$.
  \end{enumerate}
\end{lemma}
\begin{proof}
In the proposed extension of the coloring from $H_1$ to $H$, there is
no edge $e\in E_1$ where the maximum color in $e$ occurs more than
once in $e$; otherwise, the promised coloring of $H_1$ using $C_1$ is
not valid.  Consider any edge $\set{u,v,w}\in E\setminus E_1$.  By
definition of $S_1$, we have $\card{\set{u,v,w}\cap S_1}\leq 1$.  Note
that $\card{\set{u,v,w} \cap S_1}\ne 0$ as $\set{u,v,w}\not \in E_1$.
Therefore, we must have $\card{\set{u,v,w} \cap S_1}=1$.  Without loss
of generality, assume that $u\in S_1$ and $v,w\not \in S_1$.  Then, in
the proposed coloring, $c_1$ is only used for $u$, while $v,w$ are
colored using some color(s) from $C_1$.  So, $c_1$ is the largest
color in $\set{u,v,w}$ and occurs exactly once.  Hence, for every
edge, the corresponding (multi)set of colors has a unique maximum, and
we conclude that the proposed coloring is a proper LO coloring of $H$.

Similarly, in the proposed extension of coloring from $H_2$ to $H$
there is no edge $e\in E_2$ where the maximum color in $e$ occurs more
than once in $e$.  Again, consider any edge $\set{u,v,w}\in E\setminus
E_2$.  In this case, we have $\card{\set{u,v,w}\cap S_2}= 2$.  Without
loss of generality, assume that $u,v\in S_2$ and $w\not \in S_2$.
Then, in the proposed coloring, $c_2$ is only used on $u,v$, while $w$
is colored using some color $c$ from $C_2$.  So, $c$ is the largest
color in $\set{u,v,w}$ and it occurs exactly once.  Hence, for every
edge, the corresponding (multi)set of color has a unique maximum, and
the proposed coloring is therefore a proper LO coloring of $H$.
\end{proof}

The following proposition is essentially Lemma 1 in \cite{blum1994new} and follows in a straight-forward manner from \autoref{lem:lo-valid-progress}.

\begin{proposition}[Proposition 5 in \cite{NZ22}]
\label{prop:blum-progress}
  Let $H=(V,E)$ be an $\eps$-balanced, 2-LO colorable 3-uniform linear
  hypergraph on $m$ vertices.
  Suppose we can find an odd independent set of size at least
  $f(m)$ in $H$ or an even
  independent set of size at least $f(m)$ in $H$ (where $f$ is {\em
    nearly-polynomial}\footnote{Definition 1 in \cite{blum1994new}.  A
  function $f(m) = m^{\alpha} \polylog{m}$ for $\alpha >0$ is nearly-polynomial. }), then there exists a
  polynomial-time algorithm that colors any 
$\eps$-balanced, 2-LO colorable 3-uniform linear
  hypergraph on $n$ vertices with $n/f(n)$ colors.
  \end{proposition}

Following this standard notion of ``making progress'' from
\cite{blum1994new}, we simply need to show that we can find an even or
an odd independent set of size at least $f(m)$ in a 2-LO colorable
3-uniform $\eps$-balanced hypergraph on $m$ vertices.  This will imply
that we can color $H_B$ with $|V_B|/f(V_B)$ colors.  We will show that
we can set $f(m) = \widetilde{\Theta}(m^{4/5})$, which will yield the
bound in \prettyref{lem:colorbalanced}.

As is typical, our coloring algorithm makes progress using two
different methods and chooses between the two methods depending on the
degree.  In the high-degree case, we use the method from \cite{NZ22}
to find a large even independent set.  The method to find a large even
independent from \cite{NZ22} requires the input hypergraph to be a
linear hypergraph, which, as discussed previously, we can assume by
\prettyref{prop:linear-reduction}.

\begin{proposition}[Proposition 11 in \cite{NZ22}]
\label{prop:high-degree-case}
Let $H=(V,E)$ be a linear 2-LO colorable 3-uniform hypergraph and
$\Delta$ be such that $\card{E}=\Omega(\Delta \card{V})$.  Then
there is a polynomial-time algorithm that finds a even independent set
of size at least $\Omega(\sqrt{|V| \Delta})$.
\end{proposition}

In the low-degree case, we show how to use an SDP based rounding
method to find a large odd independent set.  Here, we capitalize on
the assumption that our input hypergraph is $\eps$-balanced to obtain
an improvement over the analogous lemma from \cite{NZ22}.  In Section
\ref{sec:sdpRound}, we prove \autoref{lem:odd-independent-set}.
\begin{restatable}{lemma}{oddIndependentSet}
\label{lem:odd-independent-set}
Let $H=(V,E)$ be a $\frac{1}{|V|^{100}}$-balanced 2-LO colorable 3-uniform hypergraph
$H=(V,E)$ with average degree at most $\Delta$.  Then there exists a 
(randomized) polynomial-time algorithm to compute an
odd independent set of size at least
\ $\Omega\paren{\frac{|V|}{\Delta^{1/3}\paren{\ln \Delta}^{3/2}}}$.
\end{restatable}

Finally, we are now ready to prove \prettyref{lem:colorbalanced}.

\begin{proof}[Proof of \autoref{lem:colorbalanced}]
We need to show that on a linear 2-LO-colorable 3-uniform
$\eps$-balanced hypergraph on $m$ vertices, we can always find either
an even independent set or an odd independent set of size at least
$f(m) = \widetilde{\Omega}(m^{4/5})$.  By
  \autoref{lem:lo-valid-progress}, this will imply we can color $H_B$ with $|V_B|/f(|V_B|)$ colors.  

Take $\Delta$ be a parameter (fixed later) so that we say we are in
the high-degree regime if the average degree is higher than $\Delta$.
Otherwise, we say that we are in the low-degree regime.  In the high
degree-regime, use \prettyref{prop:high-degree-case} to find an even
independent set $S$ of size at least $\Omega(\sqrt{m\Delta})$.  In the
low-degree regime, we invoke \prettyref{lem:odd-independent-set} to
find an odd independent set $S$ of size at least $\widetilde
\Omega(m/\Delta^{1/3})$.  Setting $\Delta = m^{3/5}$ implies that the
independent set we find has size at least $m^{4/5}$.
Finally, by \autoref{prop:blum-progress} we have the desired bound on
the number of colors used.
\end{proof}

\section{Combinatorial Rounding for Unbalanced Vertices}\label{sec:combRound}

In this section, we prove \prettyref{lem:unbalancedPartialColoring}.
In other words, we show that for any $\eps>0$,
\prettyref{alg:combinatorial-rounding} outputs a partial LO coloring
using $O\paren{\log \paren{\frac{1}{\eps}}}$ colors so that all the
unbalanced vertices are assigned a color.

\unbalancedPartialChi*

To prove this lemma, we give an algorithm, which given the value
$\{\gamma_v\}$ for each vertex $v$ (from 
\prettyref{sdp:lo_coloring}), is then combinatorial.  
The algorithm also takes as input
the value of $\eps$, which is the parameter
we use to define $\eps$-balanced.

\begin{algorithm}
\caption{Combinatorial Rounding}
\label{alg:combinatorial-rounding}
Input: {A 2-LO colorable 3-uniform hypergraph $H=(V, E)$, $\eps > 0$, the
values $\{\gamma_a\}$ for all $a \in V$ and set $C$ of linearly
ordered colors.}

Output: A partial LO coloring of all unbalanced vertices in $V$.
\begin{enumerate}
\item Set $j \assign 0, \ell_0 \assign -1, u_0 \assign 1, I_0 \assign [\ell_0,u_0]$.
\item \label{step:preprocessing} While $I_j\nsubseteq [-1/3-\eps,
  -1/3+\eps]$ do:
\begin{enumerate}
    \item If $j$ is even then set $I_{j+1}$ to the lower half of
      $I_j$, if $j$ is odd then set $I_{j+1}$ to be the upper half of
      $I_j$.  More precisely, set 
    \begin{align*}
        \ell_{j+1} \assign 
        \begin{cases}
        \frac{\ell_j+u_j}{2}&j\text{ is odd}\\
        \ell_j&\text{otherwise}
        \end{cases}
        &&
        u_{j+1} \assign 
        \begin{cases}
        \frac{\ell_j+u_j}{2}&j\text{ is even}\\
        u_j&\text{otherwise}
        \end{cases}
    \end{align*}
    and set $I_{j+1} \assign [\ell_{j+1},u_{j+1}]$.
    \item Set $S_{j+1} \assign \set{a\in V| \gamma_a\in I_j\setminus I_{j+1}}$ and color $S_{j+1}$ using the largest unused color from $C$.
    \item Set $j \assign j+1$.
\end{enumerate}

\end{enumerate}
\end{algorithm}

We will use the following observation.
\begin{observation}
\label{obs:vect-prop}
For any $\set{a,b,c} \in E$, we have
$ \gamma_a + \gamma_b + \gamma_c = -1 $. 
\end{observation}
\begin{proof}
From constraint \prettyref{eq:sdp-sum}, we get
$ \gamma_a + \gamma_b + \gamma_c = \inprod{\vect{a} + \vect{b} + \vect{c}, \vectone} = \inprod{-\vectone,\vectone} = - 1.$\end{proof}

On a high level, the algorithm partitions the interval $[-1,1]$ and
assigns colors to vertices depending on where their corresponding
$\gamma_a$ values fall in this interval.  For example, in the first
iteration of the algorithm, we set $S_1$ to contain all vertices whose
$\gamma_a$ values fall into the interval $(0,1]$.  Notice that 
by \prettyref{obs:vect-prop}, at most one vertex from an edge will
qualify.  Now, all remaining vertices have $\gamma_a$ values in the
interval $[-1,0]$.
Next, we consider all vertices whose $\gamma_a$ values fall
into the interval $[-1,-1/2)$.  Again, an edge with all three values
  in $[-1,0]$ can not have more than one vertex with $\gamma_a$ value
  in $[-1,-1/2)$, and so on.  We now formally analyze the algorithm.

\begin{lemma}\label{lem:bounds}
For even $j\geq 2$, the interval $[\ell_j, u_j]$ is
$$ \left[\frac{-(2^{j-1}-2)/3-1}{2^{j-1}}, \frac{-(2^{j-1}-2)/3}{2^{j-1}} \right].$$
For odd $j\geq 1$, the interval $[\ell_j, u_j]$ is
$$ \left[\frac{-(2^{j-1}-1)/3-1}{2^{j-1}}, \frac{-(2^{j-1}-1)/3}{2^{j-1}} \right].$$
\end{lemma}

\begin{proof}
For $j=1$ the interval is $[-1,0]$ and $j=2$ the interval is $[-1/2,0]$.  For odd $j$, we have
$$
\ell_{j+1} = \frac{\ell_j + u_j}{2} =  \frac{-(2^{j-1}-1)/3-1 -(2^{j-1}-1)/3}{2^j} \\
=  \frac{-(2^{j}-2)/3-1}{2^j},
$$
and
$$ u_{j+1} = \frac{-(2^{j-1}-1)/3}{2^{j-1}} = \frac{-(2^{j}-2)/3}{2^{j}}.$$
For even $j$, we have

$$
u_{j+1} = \frac{\ell_j + u_j}{2} =  \frac{-(2^{j-1}-2)/3-1 -(2^{j-1}-2)/3}{2^j} \\
=  \frac{-(2^{j}-1)/3}{2^j},$$
and
$$ \ell_{j+1} = \frac{-(2^{j-1}-2)/3-1}{2^{j-1}} = \frac{-(2^{j}-4)/3-2}{2^{j}} =
\frac{(-2^{j}+1-3)/3}{2^{j}} = \frac{-(2^{j}-1)/3-1}{2^{j}}
.$$\end{proof}

As a consequence of \autoref{lem:bounds}, we immediately get a bound on the number of iterations in form of \autoref{cor:preprocessing-time}.

\begin{corollary}\label{cor:preprocessing-time}
For $j\geq \log{(\frac{4}{3\eps})}$, we have $I_j\subseteq [-1/3-\eps, -1/3+\eps]$.
\end{corollary}

\begin{proof}
By \autoref{lem:bounds} we have the following bounds on $I_j$.
For even $j\geq 2$, the interval $I_j = [\ell_j, u_j]$ is
$$ \left[-\frac{1}{3} - \frac{1}{3\cdot 2^{j-1}}, -\frac{1}{3} + \frac{2}{3 \cdot 2^{j-1}}\right].$$
For odd $j\geq 1$, the interval $I_j = [\ell_j, u_j]$ is
$$ \left[-\frac{1}{3} - \frac{2}{3\cdot 2^{j-1}}, -\frac{1}{3} + \frac{1}{3 \cdot 2^{j-1}}\right].$$
Setting $\eps = \frac{1}{3\cdot 2^{j-2}}= \frac{4}{3 \cdot 2^j}$, we have $I_j \subseteq [-\frac{1}{3} - \eps, - \frac{1}{3} + \eps]$.  Thus, $j = \log{(\frac{4}{3 \eps})}$.
\end{proof}

In \autoref{lem:reduction-independence} we show that in each iteration \prettyref{alg:combinatorial-rounding} colors an odd independent set.
\autoref{lem:reduction-independence} also follows from \autoref{lem:bounds}.
\begin{lemma}\label{lem:reduction-independence}
For each $j\geq 0$, let $H_j=(S_j,E_j)$ be a hypergraph with $E_j =
\set{e \in E : e \subseteq S_j}$.  Then for any $j\geq 0$, the set
$S_{j+1}$ is an odd independent set (i.e., we have $\card{S_{j+1}\cap
  e}\leq 1$ for any $e\in E_j$).
\end{lemma}
\begin{proof}
Let $\set{a,b,c}\in E_j$. Suppose $a, b \in S_{j+1}$.
If $j$ is odd, then $\gamma_a, \gamma_b\in [\ell_j, \ell_{j+1})$. Therefore, we get
$\gamma_a+\gamma_b<2 \ell_{j+1}$.
This implies, by \prettyref{obs:vect-prop}, $\gamma_c > -1 - 2 \ell_{j+1}$.
Therefore, we have
\begin{align*}
\gamma_c > -1-2 \ell_{j+1} &= -1 -2\left(\frac{-(2^{j}-2)/3-1}{2^j}\right) =
-1 +2\left(\frac{(2^{j}-2)/3+1}{2^j}\right) \\
& = \frac{-3 \cdot 2^j +2(2^{j}-2)+6}{3 \cdot 2^j} =
\frac{-2^j +2}{3 \cdot 2^j} = u_j,
\end{align*}
which is a contradiction since $\gamma_c \in [\ell_j, u_j]$.

Similarly, if $j$ is even, then $\gamma_a, \gamma_b\in (u_{j+1}, u_{j}]$ as $a,b \in S_{j+1}$.
Therefore, we get
$\gamma_a+\gamma_b>2u_{j+1}$.
This implies, by \prettyref{obs:vect-prop}, $\gamma_c < -1 - 2u_{j+1}$.
Therefore, we have 
\begin{align*}
\gamma_c < -1-2u_{j+1} &= \frac{- 3 \cdot 2^j + 2 \cdot 2^j-2}{3 \cdot 2^j} =
\frac{- 2^j -2}{3 \cdot 2^j} = \ell_j, 
\end{align*}
which is again a contradiction to the fact that $\gamma_c \in [\ell_j, u_j]$.
\end{proof}

\begin{proof}[Proof of \autoref{lem:unbalancedPartialColoring}]
By \prettyref{cor:preprocessing-time},
\prettyref{alg:combinatorial-rounding} runs for
$O\paren{\log\paren{\frac{1}{\eps}}}$ iterations.  In each iteration,
it uses exactly one color, which yields the stated bound on the number
of colors used.
To show that the output coloring is a partial LO coloring, we apply
\prettyref{lem:reduction-independence}, which states that each color
corresponds to an odd independent set.

Now, we need to show that any
edge with at least one colored vertex will have a unique maximum color.
Consider such an edge $e=\set{a, b, c}$.
If only one vertex in $e$ is colored, then we are done.
First, assume exactly two vertices in $e$ (say $a,b$) were colored.
Let $a$ be colored in the  $j_a$-th iteration and $b$ be colored in the $j_b$-th iteration.
Assume (without loss of generality) that $j_a\geq j_b$. Then, by
\autoref{lem:reduction-independence} we have $j_a\ne j_b$ (i.e., $j_a>j_b$).
As the color used in the iteration $j$ is the $j$-th largest color in $C$ (by a simple induction) color assigned to $a$ is strictly larger than the color assigned to $b$.
Finally, if all the vertices, $a, b, c$ were colored at iterations $j_a\geq j_b \geq j_c $, respectively.
Then, again by the same arguments we have $j_a>j_b$ and $j_a>j_c$, so the maximum color is assigned to only $a$.
\end{proof}

\section{SDP Rounding for Balanced Hypergraphs}\label{sec:sdpRound}

In this section we 
show that \prettyref{alg:sdp-random-projection}
outputs an odd independent set in an $\eps$-balanced 2-LO colorable
3-uniform hypergraph $H_B = (V_B, E_B)$.
Thus, we will prove \autoref{lem:odd-independent-set}.
\oddIndependentSet*

Recall that we have a solution for the SDP \ref{sdp:lo_coloring}.  Let
$\vectu{a}$ be the unit vector along the component orthogonal to
$\vect{\emptyset}$ (if the orthogonal component is zero, then we
define $\vectu{a}$ to be any arbitrarily chosen unit
vector). Therefore,
\begin{equation}
\vectu{a} = \frac{\vect{a} - \gamma_a \vectone}
	{\norm{\vect{a} - \gamma_a \vectone}}
	=  \frac{\vect{a} - \gamma_a\vectone}
	{\sqrt{\norm{\vect{a}}^2 + \gamma_a^2 - 2 \gamma_a \inprod{\vect{a},\vectone}}}
	= \frac{\vect{a} - \gamma_a\vectone}
	{\sqrt{1 - \gamma_a^2}}.
\end{equation}
Let function $\bar{\Phi} : \R \to [0,1]$ be defined as $\gcap{t} \defeq 
\Prob{g \sim \cN(0,1)}{g \geq t}$.

\begin{algorithm}[H]
\caption{Randomized Rounding}
\label{alg:sdp-random-projection}
Input: $H_B$ a $\eps$-balanced 2-LO colorable 3-uniform hypergraph and
a parameter $\alpha$ (see \prettyref{lem:independent-set-size} for
values of $\eps$ and $\alpha$ to be used).

Output: An odd independent set.
\begin{enumerate}
	\item Let $t$ be such that $\alpha = \gcap{t}$.
	\item \label{step:rounding_step} Sample $g \sim \cN\paren{0,1}^{|V_B|}$ and set
	$S(t) := \set{a \in V_B : \inprod{\vectu{a},g} \geq t}$.
	\item Set $S'(t) := S(t) \setminus \paren{\displaystyle \bigcup_{\substack{e \in E_B \\ \Abs{e \cap S(t)} \geq 2 }} e}$.
	\item Output $S'(t)$. 
\end{enumerate}

\end{algorithm}

In case of an edge $\set{a, b, c}$ with perfectly balanced vertices
(i.e., if we have $\gamma_a=\gamma_b=\gamma_c=-1/3$), one can observe
that the component orthogonal to $\vect{\emptyset}$ of the 
corresponding vectors sum to 0 (i.e., we have
$\vectu{a}+\vectu{b}+\vectu{c}=0$).  In \autoref{lem:u_bound} we show
a generalization of this observation for an $\eps$-balanced
hypergraph.  Recall that in an $\eps$-balanced hypergraph, we have
$\gamma_a \in [-1/3-\eps, -1/3+\eps]$ for each vertex.
The proof of the next lemma can be found in \prettyref{app:omitted}.
\begin{lemma}
\label{lem:u_bound}
Let $\set{a, b, c}$ be an edge in an $\eps$-balanced hypergraph $H_B$.
Then $\norm{\vectu{a}+\vectu{b}+\vectu{c}}^2\leq 18\eps$.
\end{lemma}

When all the vertices in $\set{a, b, c}$ are perfectly balanced
then the event that both $a$ and $b$ belong to $S(t)$ is equivalent to 
$\inprod{\vectu{c}, g}\leq -2t$ as $\vectu{a}+\vectu{b}+\vectu{c}=0$.
Therefore, we can use bounds on Gaussians to bound the probability of the aforementioned event.
Again, \autoref{lem:correlation-bound} generalizes this to $\eps$-balanced vector for small enough $\eps$.
\begin{lemma}
\label{lem:correlation-bound}
Take $\eps=\frac{1}{|V_B|^{100}}$ and
let $a,b$ be adjacent vertices in $H_B$.
Then 
\[
\Prob{}{a \in S(t) \land b \in S(t)} \leq \gcap{2t} + \frac{2}{|V_B|^{25}}.
\]
\end{lemma}

\begin{proof}
Suppose $e = \set{a,b,c}$ is an edge in $H_B$ containing both $a$ and $b$.
If both $a$ and $b$ belong to $S(t)$, then $\inprod{\vectu{a},g} \geq t$
and $\inprod{\vectu{b},g} \geq t$.
Note that $\norm{\vectu{a} + \vectu{b} + \vectu{c}} \leq 3\sqrt{2\eps}$ by \prettyref{lem:u_bound}.
If we additionally assume that $\norm{g}\leq |V_B|^{25}$ (this assumption is violated with low probability) we have
\begin{align*}
	3\sqrt{2\eps}|V_B|^{25} &\geq \inprod{\vectu{a}+\vectu{b}+\vectu{c},g}&\text{(Cauchy-Schwarz)}\\
	&= \inprod{\vectu{a},g}+\inprod{\vectu{b}, g}+\inprod{\vectu{c}, g}\\
	&\geq 2t+\inprod{\vectu{c}, g}&\paren{\inprod{\vectu{a}, g}\geq t \text{ and }\inprod{\vectu{b}, g}\geq t}
\end{align*}
we get $\inprod{\vectu{c}, g}\leq -2t + 3\sqrt{2\eps}|V_B|^{25}$.
Thus,
we can upper bound $\Prob{}{\paren{\inprod{\vectu{a},g} \geq t} \land \paren{\inprod{\vectu{b},g} \geq t} \land \paren{\norm{g}\leq |V_B|^{25}}}$
by
\begin{align*}
&\Prob{}{\paren{\inprod{\vectu{c},g} \leq -2t+3\sqrt{2\eps}|V_B|^{25}} \land \paren{\norm{g}\leq |V_B|^{25}}} \\
&\leq \Prob{}{\inprod{\vectu{c},g} \leq -2t+3\sqrt{2\eps}|V_B|^{25}} \\
&= \gcap{2t-3\sqrt{2\eps}|V_B|^{25}}\\
&\leq \gcap{2t}+\sqrt{\eps}|V_B|^{25}&(\text{\prettyref{fact:gaussian-concentration}})\\
&\leq \gcap{2t}+\frac{1}{|V_B|^{25}}.&\paren{\eps=\frac{1}{|V_B|^{100}}}
\end{align*}
Now, in the following step we look at the case when the assumption $\norm{g}\leq |V_B|^{25}$ is violated.
\begin{align*}
\Prob{}{\paren{\inprod{\vectu{a},g} \geq t} \land \paren{\inprod{\vectu{b},g} \geq t} \land \paren{\norm{g}> |V_B|^{25}}}&\leq \Prob{}{\norm{g}> |V_B|^{25}}\\
&\leq \Prob{}{\norm{g}^2 > |V_B|^{50}}\\
&\leq \frac{1}{|V_B|^{49}}&\paren{\Ex{}{\norm{g}^2}=|V_B|\text{ and Markov bound}}
\end{align*}
Adding up the two disjoint cases we get the required bound.
\end{proof}

\begin{lemma}
\label{lem:independent-set-size}
Let $\Delta\geq 4$ be an upper bound on the average degree of a vertex in $H_B$ (i.e.,
$\card{E_B}\leq \frac{\Delta \card{V_B}}{3}$).
Take $\alpha =
\frac{1}{32} \frac{1}{\Delta^{\frac{1}{3}}(\ln \Delta)^{1/2}}$
and $\eps = \frac{1}{|V_B|^{100}}$. 
Then, we have
$
  \Ex{}{|S'(t)|} \geq \frac{3}{4}\alpha \card{V_B}.
$
\end{lemma}

\begin{proof}
To lower bound the expected size of $S'(t)$ we lower bound the expected size of $S(t)$
and upper bound the expected number of vertices participating
in a bad edge (i.e., an edge $e$ such that $|e\cap S(t)|\geq 2$) separately.

First, we lower bound the size of $|S(t)|$ as follows. 
\begin{align*}
\Ex{}{\card{S(t)}}=\sum_{a\in V_B}\Prob{}{\inprod{\vectu{a},g}\geq t}=\alpha \card{V_B}.
\end{align*}

Now, to get an upper bound we note that each bad edge can contribute at most 3 vertices in the total number of vertices participating in some bad edge.
Formally, we have the following. 
\begin{align*}
\card{\bigcup_{\substack{e \in E_B \\ \Abs{e \cap S(t)} \geq 2 }}e} &\leq \sum_{\substack{e \in E_B \\ \card{e \cap S(t)} \geq 2 }}|e|&\text{(Union Bound)}\\
&\leq 3\card{\set{e \in E_B \text{ s.t. }\card{e \cap S(t)} \geq 2 }}&(|e|=3)
\end{align*}

If an edge $\set{a,b,c}$ is bad, i.e., we have $|\set{a, b, c}\cap S(t)|\geq 2$, then either $\set{a, b}\subseteq S(t)$ or $\set{a, c}\subseteq S(t)$ or $\set{b, c}\subseteq S(t)$.
Therefore, by union bound $\Ex{}{\card{\set{e \in E_B \text{ s.t. }\card{e \cap S(t)} \geq 2 }}}$ is at most
\begin{align*}
& \sum_{\set{a,b,c}\in E_B}
\paren{\Prob{}{a\in S(t) \land b\in S(t)} + \Prob{}{a\in S(t) \land c\in S(t)} + \Prob{}{c\in S(t) \land b\in S(t)}}\\
& \leq \sum_{e\in E_B} 3\cdot \paren{\gcap{2t} + \frac{2}{|V_B|^{25}}}\\
& = 3 |E_B| \cdot \gcap{2t}+\frac{6|E_B|}{|V_B|^{25}},
\end{align*}
where the second inequality follows from \autoref{lem:correlation-bound}.
Let us now upper bound the first term as follows.
\begin{align*}
	3|E_B|\gcap{2t} & \leq \Delta |V_B| \cdot 512 \gcap{t}^{4} \cdot \paren{\ln({1/\gcap{t}})}^{3/2}&\paren{|E_B|\leq \frac{\Delta |V_B|}{3}\text{ and \prettyref{cor:gcap-twice-t-bound}}}\\
	& \leq \Delta |V_B| \cdot 512 \alpha^{4} \cdot \paren{\ln({1/\alpha})}^{3/2}&(\gcap{t}=\alpha)\\
	&\leq \frac{1}{8}\alpha|V_B|\paren{\frac{\ln (1/\alpha)}{4\ln \Delta}}^{3/2} &(\text{Substituting }\alpha)\\
	&\leq \frac{1}{8}\alpha |V_B|&(\Delta\geq 4)
\end{align*}
Note that in the first inequality above we could use 
\autoref{cor:gcap-twice-t-bound} as $\Delta\geq 4$ implies
$t\geq 1$.
It is easy to show that $\frac{6|E_B|}{|V_B|^{25}}\leq \frac{1}{8}\alpha|V_B|$.
Therefore, we get 
\begin{align*}
\Ex{}{\card{\set{e \in E_B \text{ s.t. }\card{e \cap S(t)} \geq 2 }}} \leq \frac{1}{4} \alpha |V_B|.
\end{align*}
Thus, by combining the two bounds we get that
\begin{equation*}
\Ex{}{\card{S'(t)}} = \Ex{}{\card{S(t)}} - \Ex{}{\card{\bigcup_{\substack{e \in E_B \\ \Abs{e \cap S(t)} \geq 2 }}e}}
 \geq \paren{1-\frac{1}{4}}\alpha \card{V_B} \geq \frac{3}{4}\alpha \card{V_B}.
\end{equation*}
\end{proof}

\begin{proof}[Proof of \autoref{lem:odd-independent-set}]
This follows from \prettyref{lem:independent-set-size} and 
the proof is standard Markov bound followed by an 
amplification argument where you repeat 
\prettyref{alg:sdp-random-projection} polynomially many times
and choose the best odd independent set among all repetitions.
The probability of even the best odd independent set not being
of the required size is then inverse exponential with respect
number of iterations.
We refer the reader to Section 13.2 of \cite{WS11} for further reference.
\end{proof}

\paragraph*{A Better SDP Rounding.}
Here we note that there is in fact a better way to round the SDP in the balanced case, which follows from \cite{hmnz24} and essentially reduces the balanced case to the unbalanced case.  Let $H = (V,E)$ be an $\eps$-balanced hypergraph on $n$ vertices.  Recall that $\eps \leq 1/n^{100}$.

As in \prettyref{alg:sdp-random-projection}, we sample a gaussian $g
\sim \cN\paren{0,1}^{n}$.  For each (unit) vector $\vectu{a}$ for $a
\in V$, let $\zeta_a = \inprod{\vectu{a}, g}$. Observe that $|\zeta_a| \in [1/n^2, n^2]$ with probability $1-O(1/n^2)$. 
Now for all $a \in V$, set $\zeta_a' = \zeta_a/n^2$ and set $\gamma_a'
= \gamma_a + \zeta'_a$.
Thus, with 
probability at
least (roughly) $1-O(1/n)$, for {\em all} vertices $a \in V$, we have $$|\zeta'_a| \in
[1/n^4, 1]~ \text{ and }~ \gamma'_a \notin (-1/3- 1/n^{100}, -1/3 + 1/n^{100}).$$
Since for every hyperedge $\set{a,b,c} \in E$, we have $\zeta'_a +
\zeta'_b + \zeta'_c = 0$ (because $\vectu{a} + \vectu{b} + \vectu{c} =
     {\bf{0}}$, which implies $\inprod{\vectu{a},g} +
     \inprod{\vectu{b},g} + \inprod{\vectu{c},g} = 0$).

Then for every hyperedge $\set{a,b,c} \in E$,
we have $\gamma_a' + \gamma_b' + \gamma_c' = -1$.  
Thus, we can run
\prettyref{alg:combinatorial-rounding} on the inputs $\{\gamma_a'\}_{a \in V}$
and $\eps = 1/n^{100}$.  By
\prettyref{lem:unbalancedPartialColoring}, it will output an
LO-coloring of $H$ using at most $O(\log{\frac{1}{\eps}})$ colors.

\section{Conclusion}

We have presented an improved bound on the number of colors needed to
efficiently LO color a 2-LO colorable 3-uniform hypergraph, and
demonstrated that SDP-based rounding methods can indeed be applied to
LO coloring.  A natural question is if we can do better than $O(\log{n})$ colors in the balanced case; this might be a step towards improving on the bound of $O(\log{n})$ colors for the general case given in \cite{hmnz24}.


\bibliographystyle{alpha}
\bibliography{ref.bib}

\appendix

\section{Properties of Gaussian}
Let function $\Phi : \R \to [0,1]$ be defined as $\ggcap{t} \defeq 
\Prob{g \sim \cN(0,1)}{g \leq t}$, and
let function $\bar{\Phi} : \R \to [0,1]$ be defined as $\gcap{t} \defeq 
\Prob{g \sim \cN(0,1)}{g \geq t}$.

\begin{fact}
	\label{fact:gaussian-concentration}
	For any $a\leq b$, we have $\gcap{b}-\gcap{a}=\Prob{g\sim \cN(0,1)}{g\in [a, b]}\leq \frac{b-a}{\sqrt{2\pi}}$.
\end{fact}
\begin{proof}
	The statement follows from the following computations.
	\begin{align*}
		\Prob{}{g\in [a,b]} = \int_{a}^{b}\frac{e^{-x^2/2}}{\sqrt{2\pi}} \, dx\leq \frac{1}{\sqrt{2\pi}}\int_{a}^{b}\sup_{y\in \mathbb R}e^{-y^2/2}\, dx=\frac{b-a}{\sqrt{2\pi}}.
	\end{align*}
\end{proof}

\begin{fact}[Folklore]
\label{fact:gcap}
For every $t > 0$,
\[ \frac{t}{\sqrt{2 \pi } (t^2 + 1) } e^{- \frac{1}{2} t^2} < \gcap{t} 
	< \frac{1}{\sqrt{2 \pi }t} e^{- \frac{1}{2} t^2} . \]
\end{fact}
\begin{corollary}[Folklore]
\label{cor:inv-gcap-log}
Fix $t\geq 1$ and let $\beta = \gcap{t}$.  Then we have 
\begin{equation*}
    \sqrt{2\ln \frac{1}{\beta}- \ln \ln \frac{1}{\beta}-\ln 16\pi}\leq t \leq \sqrt{2\ln \frac{1}{\beta}- \ln \ln \frac{1}{\beta}}\leq \sqrt{2\ln \frac{1}{\beta}}.
\end{equation*}
In fact, $t < \sqrt{2\ln \frac{1}{\beta}}$ holds even if $t\in (0,1)$.
\end{corollary}
\begin{proof}
  Let $t > 0$ (note here we allow $t\in (0,1)$) and let $\beta = \gcap{t}$.
By taking logarithm and multiplying by $-2$, the inequalities in \prettyref{fact:gcap} imply
\begin{align}
2\ln \paren{\frac{1}{\beta}} &> t^2 + 2\ln\paren{\sqrt{2\pi} t}, \label{eq:loglower}\\
2\ln \paren{\frac{1}{\beta}} &< t^2 + 2\ln\paren{\sqrt{2\pi} \paren{\frac{t^2+1}{t}}}\label{eq:logupper}.
\end{align}
We can now use \prettyref{eq:loglower} to get
\begin{align*}
2\ln \paren{\frac{1}{\beta}} &> t^2 + 2\ln\paren{\sqrt{2\pi} t} \geq t^2.
\end{align*}
Hence, we have $t < \sqrt{2\ln \frac{1}{\beta}}$ for any $t > 0$.
Again, by multiplying by $\frac{1}{2}$ and taking logarithms, the \prettyref{eq:loglower}, \prettyref{eq:logupper} imply
\begin{align}
\ln \ln \paren{\frac{1}{\beta}} &> \ln \paren{t^2/2 + \ln\paren{\sqrt{2\pi} t}}, \label{eq:logloglower}\\
\ln \ln \paren{\frac{1}{\beta}} &< \ln \paren{t^2/2 + \ln\paren{\sqrt{2\pi} \paren{\frac{t^2+1}{t}}}}. \label{eq:loglogupper}
\end{align}
From hereon we assume $t\geq 1$. \prettyref{eq:loglower} $-$ \prettyref{eq:loglogupper} gives us
\begin{equation}
\label{eq:pre-upper}
2\ln \paren{\frac{1}{\beta}} - \ln \ln \paren{\frac{1}{\beta}} > t^2 +
\ln\paren{
    \frac{2\pi t^2}
    {
        t^2/2 + \ln\paren{\sqrt{2\pi} \paren{\frac{t^2+1}{t}}}
    }
}
=
t^2 +
\ln\paren{
    \frac{4\pi t^2}
    {
        t^2 + 2\ln\paren{\sqrt{2\pi} \paren{\frac{t^2+1}{t}}}
    }
}
\end{equation}
\begin{claim}
$4\pi t^2 \geq t^2 + 2\ln\paren{\sqrt{2\pi} \paren{\frac{t^2+1}{t}}}$.
\end{claim}
\begin{proof}
Note that the above inequality is equivalent to $\paren{\frac{4\pi-1}{2}} t^2 \geq \ln \sqrt{2\pi} + \ln \paren{t+\frac{1}{t}}$.
Indeed we have
\begin{align*}
\ln\sqrt{2\pi} + \ln \paren{t+\frac{1}{t}}&\leq \ln\sqrt{2\pi} + \ln (t+1)&(t\geq 1)\\
&\leq \ln \sqrt{2\pi} + t &(\ln (1+x)\leq x)\\
&\leq \ln \sqrt{2\pi} + t^2&(t\geq 1)\\
&\leq \paren{\frac{4\pi-3}{2}} + t^2& \paren{\frac{4\pi - 3}{2}\geq \ln \sqrt{2\pi}}\\
&\leq \paren{\frac{4\pi-1}{2}}t^2& \paren{t\geq 1}
\end{align*}
\end{proof}
Using this claim and \prettyref{eq:pre-upper} we get
\begin{equation*}
t^2 \leq 2\ln \frac{1}{\beta}- \ln \ln \frac{1}{\beta}.
\end{equation*}
Hence, we have $t \leq \sqrt{2\ln \frac{1}{\beta}- \ln \ln \frac{1}{\beta}}$.
For the remaining inequality, we again see that \prettyref{eq:logupper} $-$ \prettyref{eq:logloglower} gives us
\begin{align*}
2\ln \frac{1}{\beta} - \ln \ln \frac{1}{\beta} &< t^2 + \ln \paren{\frac{4\pi\paren{t+\frac{1}{t}}}{t^2+2\ln(\sqrt{2\pi}t)}}\\
&\leq t^2 + \ln 4 \pi + \ln \paren{\frac{(t+1)^2}{t^2}}&\paren{t\geq 1}\\
&\leq t^2 + \ln 4 \pi + 2\ln \paren{1+\frac{1}{t}}\\
&\leq t^2 + \ln 4 \pi + 2\ln 2 = t^2 + \ln 16 \pi&(t\geq 1)
\end{align*}
$\sqrt{2\ln \frac{1}{\beta}- \ln \ln \frac{1}{\beta}-\ln 16\pi}\leq t$ follows from the above inequality.
Hence, we have all the required inequalities.
\end{proof}

\begin{corollary}[Folklore]
	\label{cor:gcap-twice-t-bound}
	Fix $t\geq 1$. Then, we have
	\begin{align*}
		\gcap{2t}\leq 512\paren{\ln \paren{\frac{1}{\gcap{t}}}}^{3/2}\gcap{t}^{4}.
	\end{align*}
\end{corollary}

\begin{proof}
\label{app:proof-of-gcap-twice-t-bound}
For any $t\geq 1$ and $\delta \in (0,1)$ the following holds.
\begin{align*}
		\gcap{2t}&\leq \frac{1}{2\sqrt{2\pi}t}e^{-2t^2}&\paren{\text{\prettyref{fact:gcap}}}\\
		&\leq \frac{1}{2\sqrt{2\pi}t} \cdot \frac{{(2\pi)}^{2} \paren{t^2+1}^{4}}{t^{4}} \cdot \gcap{t}^{4}&\paren{\frac{t}{\sqrt{2 \pi } (t^2 + 1) } e^{- \frac{1}{2} t^2} \leq \gcap{t} \text{ by \prettyref{fact:gcap}}}\\
		&= (2\pi)^{3/2} \frac{1}{2t} \paren{t+\frac{1}{t}}^4   \gcap{t}^{4}\\
		&\leq (2\pi)^{3/2} (2t)^{3}   \gcap{t}^{4}&\paren{t\geq 1}\\
		&\leq (4\sqrt{\pi})^3\paren{ \ln \paren{ \frac{1}{\gcap{t}} } }^{3} \cdot \gcap{t}^{4}&(\text{by \prettyref{cor:inv-gcap-log}})\\
		&\leq 512\paren{\ln \paren{\frac{1}{\gcap{t}}}}^{3}\gcap{t}^{4}&(\sqrt{\pi}\leq 2).
	\end{align*}
\end{proof}

\section{Omitted Proofs}
\label{app:omitted}

\subsection{Proof of \prettyref{lem:u_bound}}
\label{app:proof-of-u-bound}
Before we proceed to prove \prettyref{lem:u_bound} we need the following lemma.
\begin{lemma}
\label{lem:eps_vectors}
Let $\set{a,b,c}\in E$ and $\gamma_a=-1/3+\eps_a$, $\gamma_b=-1/3+\eps_b$, $\gamma_c=-1/3+\eps_c$.
Then the following hold.
\begin{enumerate}
    \item \label{item:eps_sum} $\eps_a+\eps_b+\eps_c = 0$.
    \item \label{itm:innerprod} $\inprod{\vect{a}, \vect{b}} = -1/3+\eps_c$, $\inprod{\vect{b}, \vect{c}} = -1/3+\eps_a$, and $\inprod{\vect{c}, \vect{a}} = -1/3+\eps_b$.
\end{enumerate}
\end{lemma}
\begin{proof}
Using \prettyref{obs:vect-prop} we get
\begin{align*}
-1/3+\eps_a -1/3+\eps_b -1/3+\eps_c=-1,
\end{align*}
which implies $\eps_a+\eps_b+\eps_c = 0$.
Taking inner products with $\vect{a}, \vect{b}, \vect{c}$ on both sides of constraint \prettyref{eq:sdp-sum} of \prettyref{sdp:lo_coloring} we get
\begin{align*}
1 + \inprod{\vect{a},\vect{b}} + \inprod{\vect{a},\vect{c}} & = 1/3 - \eps_a,\\
\inprod{\vect{b},\vect{a}} + 1 + \inprod{\vect{b},\vect{c}} & = 1/3 - \eps_b,\\
\inprod{\vect{c},\vect{a}} + \inprod{\vect{c},\vect{b}} + 1 & = 1/3 - \eps_c,
\end{align*}
which imply
\begin{align}
\inprod{\vect{a},\vect{b}} + \inprod{\vect{a},\vect{c}} & = - (2/3 + \eps_a), \label{eq:epsa}\\
\inprod{\vect{b},\vect{a}} + \inprod{\vect{b},\vect{c}} & = - (2/3 + \eps_b), \label{eq:epsb}\\
\inprod{\vect{c},\vect{a}} + \inprod{\vect{c},\vect{b}} & = - (2/3 + \eps_c). \label{eq:epsc}
\end{align}
\prettyref{eq:epsa}$+$\prettyref{eq:epsb}$-$\prettyref{eq:epsc} gives us
\begin{equation*}
2\inprod{\vect{a},\vect{b}} = -2/3 - \eps_a - \eps_b + \eps_c.
\end{equation*}
Using \prettyref{item:eps_sum} of this lemma and dividing by 2 we get $\inprod{\vect{a}, \vect{b}} = -1/3+\eps_c$ as needed.
Similarly, we get $\inprod{\vect{a}, \vect{c}}=-1/3+\eps_b$, $\inprod{\vect{c}, \vect{b}}_c=-1/3+\eps_a$.
\end{proof}
\begin{proof}[Proof of \autoref{lem:u_bound}]
Note that
\begin{align*} 
\inprod{\vectu{a},\vectu{b}} & = \inprod{\frac{\vect{a} - \gamma_a\vectone}
{\sqrt{1 - \gamma_a^2}} ,\frac{\vect{b} - \gamma_b\vectone}{\sqrt{1 - \gamma_b^2}}}
= \frac{\inprod{\vect{a},\vect{b}} - \gamma_a \inprod{\vectone,\vect{b}}
- \gamma_b \inprod{\vectone,\vect{a}} + \gamma_a \gamma_b \inprod{\vectone,\vectone}}
{\sqrt{\paren{1 - \gamma_a^2}\paren{1 - \gamma_b^2}}} \\
& = \frac{\inprod{\vect{a},\vect{b}} - \gamma_a \gamma_b}
		{\sqrt{\paren{1 - \gamma_a^2}\paren{1 - \gamma_b^2}}} .\nonumber
\end{align*}
First let us upper-bound the denominator in the above expression using $\gamma_a, \gamma_b\in [-1/3-\epsilon, -1/3+\epsilon]$ as follows.
\begin{align*}
 \sqrt{\paren{1 - \gamma_a^2}\paren{1 - \gamma_b^2}}&\leq \sqrt{(1-(1/3-\epsilon)^2)(1-(1/3-\epsilon)^2)}\\
 &= \frac{8}{9}+\frac{2\epsilon}{3} - \epsilon^2\\
 &\leq \frac{8}{9}+\frac{2\epsilon}{3}.
\end{align*}
This implies that
\begin{align*}
\frac{1}{\sqrt{\paren{1 - \gamma_a^2}\paren{1 - \gamma_b^2}}}&\geq \frac{1}{\frac{8}{9}\paren{1+\frac{9\epsilon}{4}}}\\
&\geq \frac{9}{8}\paren{1-\frac{9\epsilon}{4}+\frac{\paren{\frac{9\epsilon}{4}}^2}{\paren{1+\frac{9\epsilon}{4}}}}\\
&\geq \frac{9}{8}\paren{1-\frac{9\epsilon}{4}}.
\end{align*}
By \prettyref{lem:eps_vectors}, we have $\inprod{\vect{a},\vect{b}}\in [-1/3-\epsilon, -1/3+\epsilon]$.
So, we can also bound the numerator in the expression for $\inprod{\vectu{a},\vectu{b}}$ by using the fact that $\inprod{\vect{a},\vect{b}},\gamma_a, \gamma_b \in [-1/3-\epsilon, -1/3+\epsilon]$ as follows.
\begin{align*}
\inprod{\vect{a},\vect{b}} - \gamma_a \gamma_b &\leq -\frac{1}{3}+\epsilon - \paren{\frac{1}{3}-\epsilon}^2\\
&= -\frac{4}{9} + \frac{5\epsilon}{3} - \epsilon^2\\
&\leq -\frac{4}{9} + \frac{5\epsilon}{3}.
\end{align*}
Therefore, we get
\begin{align*}
\inprod{\vectu{a},\vectu{b}}&\leq -\frac{4}{9}\paren{1-\frac{15\epsilon}{4}}\cdot\frac{9}{8}\paren{1-\frac{9\epsilon}{4}}\\
& \leq -\frac{1}{2}\paren{1-6\epsilon}.
\end{align*}
Finally, for the edge $\set{a, b, c}$ we get
	\begin{align*}
		\norm{\vectu{a}+\vectu{b}+\vectu{c}}^2 &= \norm{\vectu{a}}^2 + \norm{\vectu{b}}^2 + \norm{\vectu{c}}^2 + 2\inprod{\vectu{a},\vectu{b}} + 2\inprod{\vectu{b},\vectu{c}} + 2\inprod{\vectu{c},\vectu{a}}\\
		&= 3 + 2\paren{\inprod{\vectu{a},\vectu{b}} + \inprod{\vectu{b},\vectu{c}} + \inprod{\vectu{c},\vectu{a}}}\\
		&\leq 18\eps
	\end{align*}
	where the last inequality follows from the fact that $\inprod{\vectu{a},\vectu{b}}$, $\inprod{\vectu{b},\vectu{c}}$, and $\inprod{\vectu{c},\vectu{a}}$ are all at most $-\frac{1}{2}+3\eps$.
\end{proof}

\end{document}